\documentclass[11pt]{article}
\usepackage{parskip}
\usepackage{palatino}
\usepackage{amsmath,amssymb,amsthm,amsfonts}
\usepackage[top=8pc,bottom=8pc,left=8pc,right=8pc]{geometry}

\usepackage{setspace}
\usepackage{graphicx}
\usepackage{wrapfig}
\usepackage{epstopdf}
\usepackage{subcaption}
\usepackage{multirow}
\usepackage{nicefrac}
\usepackage{booktabs}
\usepackage{amsmath}
\usepackage{amsthm}
\usepackage[english]{babel}
\usepackage{graphicx,epstopdf}
\RequirePackage[authoryear]{natbib}
\usepackage{hyperref}
\hypersetup{colorlinks=true,citecolor=blue}
\hypersetup{linkcolor=blue}
\usepackage{algorithm}
\usepackage[noend]{algpseudocode}
\usepackage{cleveref}

\newcommand{\be}{\begin{eqnarray}}
\newcommand{\ee}{\end{eqnarray}}
\newcommand{\ba}{\begin{eqnarray*}}
\newcommand{\ea}{\end{eqnarray*}}

\newtheorem{theorem}{Theorem}[section]

\newtheorem{lemma}[theorem]{Lemma}

\newtheorem{remark}{Remark}[section]

\usepackage{bm}

\usepackage{array}
\newcolumntype{L}[1]{>{\raggedright\let\newline\\\arraybackslash\hspace{0pt}}m{#1}}
\newcolumntype{C}[1]{>{\centering\let\newline\\\arraybackslash\hspace{0pt}}m{#1}}
\newcolumntype{R}[1]{>{\raggedleft\let\newline\\\arraybackslash\hspace{0pt}}m{#1}}

\usepackage{array}

\include{math-commands}
\usepackage{tikz}
\usepackage{mathdots}

\begin{document}

 
            
            
            
 
             

             
        




\title{Social Sensors in Epidemiological Networks via Graph Eigenvectors}
\author{Shubhajit Sen, Samhita Pal, and Srijan Sengupta \footnote{Shubhajit Sen (\url{ssen8@ncsu.edu}) and Samhita Pal (\url{spal4@ncsu.edu}) are Ph.D. students in the Department of Statistics at North Carolina State University. Srijan Sengupta (\url{ssengup2@ncsu.edu}) is an Assistant Professor in the Department of Statistics at North Carolina State University. Shubhajit and Samhita contributed equally to the manuscript.}}

\bibliographystyle{plainnat}
\date{}
\maketitle

{\textbf{Abstract:} In this paper, we consider {{epidemiological networks}} which are used for modeling the transmission of contagious diseases  through a population. 
Specifically, we study the so-called \textit{social sensors} problem: given an epidemiological network, can we find a small set of nodes such that by monitoring disease transmission on these nodes, we can get ahead of the overall epidemic in the full population? In spite of its societal relevance, there has not been much statistical work on this problem, and we aim to provide an exposition that will hopefully stimulate interest in the research community.
Furthermore,  by leveraging classical results in spectral graph theory, we propose a novel method for finding social sensors, which achieves substantial improvement over existing methods in both synthetic and real-world epidemiological networks.}
 
\newpage

\section{Introduction}

In mathematical or statistical modeling, relational structures are often represented by networks, which is a set of objects (known as vertices), and the connections between any pair of the objects (known as edges). Thus, a network is fully characterized by two sets, set of vertices, and the set of edges.
Scientific research on networks has a long and rich history, going back almost four centuries to Euler's famous paper on the Seven Bridges of K{ö}nigsberg \citep{euler1741solutio}.
Today, Network Science is a rapidly growing multidisciplinary scientific paradigm,
drawing on theory and methods from mathematics, physics, computer science and statistics, with prominent applications in social sciences, economics, psychology, political science, engineering sciences, and biological sciences \citep{watts1998collective,barabasi1999emergence,adamic2005political,albert2002statistical,girvan2002community}. 
Fittingly, the last two decades have seen tremendous progress in developing statistical inference methods for network data.
This includes extensive work on community detection \citep{bickel2009nonparametric,zhao2012consistency,rohe2011spectral}, model fitting/ selection \citep{hoff2002latent,handcock2007model,krivitsky2009representing, wang2017likelihood,yan2014model,bickel2015hypothesis}, hypothesis testing \citep{ghoshdastidar2018practical,tang2017semiparametric,tang2017nonparametric}, and anomaly detection \citep{zhao2018performance,sengupta2018anomaly,komolafe2017statistical}.

In this paper, we consider {\textit{epidemiological networks}} which are used for modeling the transmission of contagious diseases  through a population \citep{keeling2005implications, bengtsson2015using, kramer2016spatial,leitch2019toward}. Here, each node represents an individual and each edge connecting a pair of nodes represents social contact with potential for pathogen transmission. We can then model a spreading process occurring on the network where contagion moves from infected nodes to non-infected nodes, using various disease models (e.g., SIR, SIS). 
Statistical inference methods are used on such networks  to predict disease transmission, estimate the epidemic threshold, identify critical hotspots, and ascertain the effect of community structure \cite{bengtsson2015using,kramer2016spatial, Boguna2002EpidemicNetworks,Wang2003EpidemicViewpoint, Chakrabarti2008,Prakash2010,Castellano2010,Nadini2018}.


Specifically, we consider the following problem: given an epidemiological network, can we find a small set of nodes such that by monitoring disease transmission on these nodes, we can get ahead of the overall epidemic in the full population?
In this paper, we consider this question from a statistical perspective.
Following \cite{fowler} and \cite{shao2016forecasting}, we call this problem as the ``social sensors" problem, as the nodes being monitored are analogous to sensors that alert us ahead of time.

Our goal in this paper is two-fold.
First, in spite of its societal relevance, there has not been much statistical work on the ``social sensors'' problem, and we aim to provide an exposition that will hopefully stimulate interest in the research community.
Second,  by leveraging classical results in spectral graph theory, we propose a novel method for finding social sensors which is a substantial improvement over existing methods.

The rest of the paper is organized as follows.
In \cref{sec:exposition}, we provide a review of existing methods for the social sensor problem.
In \cref{sec:methodology}, we propose a new method based on spectral properties of the graph.
In \cref{sec:sim} and \cref{sec:realdata}, we report numerical results on synthetic and real-world epidemiological networks, respectively, and in \cref{sec:disc}, we conclude the paper with discussion and next steps.


\section{The social sensors problem in epidemiological networks}
\label{sec:exposition}

The basic deterministic models for the transmission of infectious disease are the compartmental models. These include a wide range of models such as \textbf{SI} (Susceptible - Infected), \textbf{SIS} (Susceptible - Infected - Susceptible), \textbf{SIR} (Susceptible - Infected - Recovered), \textbf{SEIS} (Susceptible - Exposed - Infected - Susceptible), \textbf{SEIR} (Susceptible - Exposed - Infected - Recovered) etc. The numbers of susceptible, exposed, infected and recovered individuals at time $t$ are represented respectively by $S(t)$, $E(t)$, $I(t)$, and $R(t)$. The compartmental models study the rate of change of these numbers over time, assuming linear transitions from one compartment to the other, where the transition rates are taken as model parameters.


Network analysis has been used as an
analytical tool to describe the evolution and spread of epidemics in societies. When networks are used for epidemiological
purposes, edges are included if
they describe relationships capable of permitting the transfer of infection. Such a social network is usually undirected and can be considered to be fixed or can be adapted to allow random mixing among actors to some extent. In general, a network with $n$ nodes and $m$ edges connecting the vertices is denoted by G($n$,$m$). Since computing over sets is usually inconvenient, but computing over numerical arrays is easy, we often prefer to represent graphs as matrices. To do so, we first need to fix an ordering of the nodes. Then the adjacency matrix $A$ is an $n \times n$ matrix satisfying \begin{align*}
    A_{ij} & = 
    \begin{cases}
      1, & \text{when there is an edge between nodes i and j. } , \\
      0, & \text{o.w. }.
    \end{cases}
\end{align*}
and $A$ is symmetric for undirected networks. 


One might ask whether we could leverage some information from the social network in order to predict some features of the transmission in the epidemiological network. An important problem in public health surveillance domain would be to forecast some properties of the infection curve, so that some containment policies could be taken by the authorities for prevention, or at least to have some lead time to face it.

\subsection{Monitoring the friends of randomly selected individuals}

One of the earliest attempts to solve the aforementioned problem was by \cite{fowler}. They first introduced the notion of a sensor set, i.e. a subset of the set of individuals (vertices) from the original network. Their idea was to monitor this sensor set to detect contagious outbreaks before they occur in the population at large. Now, the next problem is to come up with a feasible method to choose the sensor set. In this regard, they used the underlying social network structure. They argued that during an outbreak, nodes at the center of the network are more likely to be infected sooner. Hence, choosing central individuals as the sensor set might provide the information about the outbreak in advance. However, it might be costly, and time consuming to collect the information about the entire large network. So they proposed an alternative method that does not require to do that. They leverage an interesting property of a social network, that says, on average your friends have more friends than you do (\cite{friends1}). In more formal words, friends of a randomly selected individuals in a social networks are more central (i.e. in general have higher degrees, higher betweenness centrality etc.) than the randomly chosen individuals. So their proposed strategy was \textit{to monitor the friends nominated by the randomly chosen individuals} as the sensor set. Note that in an epidemic outbreak, these nominated friends are expected to be infected sooner. From here on we refer this method as FOS approach.

They evaluated this FOS method in a flu outbreak in Harvard College in the fall of 2009. As expected there was a shift in the S-shaped cumulative incidence curve, and the daily incidence curve, detecting a significant amount of lead time. Moreover, it was observed that the friend group exhibited higher in-degree (number of times an individual was nominated by someone as a friend), higher centrality (number of shortest paths between two nodes in the network that pass through an individual), higher coreness (number of friends an individual has when individuals with lowest degrees are iteratively removed), and lower transitivity (the probability that two of one's friends are friends with one another). Moreover, the aforementioned features were also used to construct the sensor set alternatively, but none of those parameters provided significant improvement in terms of the lead time than the method originally proposed. Rather, computation of these parameters required entire information about the network structure, in contrast with the monitoring the friends method, which requires the information on the sample collected only.


This work by \cite{fowler} was one of the earliest attempt to identify the social sensor in epidemic outbreak. Although the importance of central individuals during an outbreak was not unknown (\cite{PhysRevLett.91.247901}), the credit for introducing the notion of sensors for an early detection of an outbreak goes to them. Moreover, they applied an interesting but easily comprehensible property of a social network in order to identify the sensor group that even does not require the information about the entire network. 

Despite the merits of this method, there are certain drawbacks which one should address. First of all, this method lacks mathematical rigor. Although it describes certain important properties of a social network in terms of the degree distribution, how that helps to identify the sensors in an epidemiological network is not discussed mathematically, only intuition was provided. Apart from this, this method can't predict the lead time which could be of real importance in disease management. Moreover, this method might not always provide a lead time as noted by \cite{shao2016forecasting}. It was shown that this method works better in a star like topology compared to a network where the degree does not follow a scale-free distribution. 

\subsection{Designing Social Network Sensors for Epidemics}

\cite{shao2016forecasting} suggest another way of identifying social sensors for early detection of a contagious epidemic. They noted that networks with star-like topology where a few of the central nodes have very large degrees, perform relatively better under the `Friends of Friend' approach as this graph structure facilitates inclusion of central nodes with high degrees to form the sensor group. On the other hand, in networks where the total number of nodes is large with an average number of edges spread across the network, it is difficult for the `Friends of Friend' approach to select sensors that will represent the entire graph based only on local friend-friend information. To tackle this kind of situation, they base their sensor selection technique on the objective of choosing the smallest group $S$ so that at least some nodes in $S$ contract the disease within the first $d$ days of the outbreak with probability at least $\epsilon$. This can be done by the \textbf{PLTM} (\textit{Peak Lead Time Maximization}) method.
\begin{align}
\label{eq:pltm}
    &S = \arg \max_{S} E[t_{pk}-t_{pk}(S)]\\
    &s.t. \enskip f(S) \geq \epsilon , |S| = k
\end{align}
where $t_{pk} = \arg \max_{t}I(t)$ and $t_{pk}(S)$ denotes the time of peak of the entire network and the sensor set respectively, $f(S)$ is the probability that at least one node in $S$ is infected, assuming that the disease spread started from a random initial node. However, this optimization problem is non-submodular. \cite{leskovec2007cost} developed a greedy algorithm that adds a sensor that maximizes the marginal gain based on expected penalty reduction, where a penalty is incurred depending on the time of detection and impact on the whole network before detection. Following this method, \cite{shao2016forecasting} consider a different, although related method with the aim of reaching a sub-modular optimization problem after defining $t_{inf}(v)$ as the expected infection time for node $v$,
\begin{align}
\label{eq:mait}
    & S = \arg \min_{S} \sum_{v \in S} t_{inf}(v)/|S|\\
    &s.t. \enskip f(S) \geq \epsilon, |S| = k.
\end{align}
The second method is submodular, but non-linear and as a result existing greedy approaches for maximizing submodular functions do not work directly. The authors propose two faster greedy approaches that picks nodes in non-decreasing $t_{inf}(.)$ order until $S$ has $f(S) \geq \epsilon$, namely, \textit{Transmission Tree} (\textbf{TT}) based sensors heuristic and \textit{Dominator Tree} (\textbf{DT}) based sensor heuristic. The TT based sensor selection heuristic first generates subgraphs of the whole network (called dendrograms) that contain infected nodes and edges through which the disease is transmitted and the depth of each node ($v$) is computed if the node gets infected in a dendrogram. This is done for all the generated dendrograms and the average of all such depths gives $t_{inf}(v)$. The heuristic then discards nodes for which $t_{inf}(v)$ is smaller than a specified value and from among the rest selects the first $k$ nodes with smallest $t_{inf}$ values. The DT based sensor selection heuristic, on the other hand, follows the exact same steps, except that the average depth of a node $v$ is now computed from a dominator tree generated from each dendrogram, where a node $x$ is said to dominate node $y$ in a directed graph if and only if all paths from the source node of infection to node $y$ has to pass through $x$.

Experimental studies on a star-like network of Oregon route-views and social contact networks for six large cities in the US show that the TT and DT approaches perform quite better than the algorithm proposed in \cite{fowler}, but the `Friends of friend' approach still works better in the Oregon network. Moreover, for the TT and DT approaches, observing the whole network is essential for selecting the sensor group that gives substantial lead time in detection. Although, \cite{shao2016forecasting} also do not give an estimate of the lead time, however, they empirically show the stability of the lead time with increasing monitoring days. Also, their methods have high variance in lead time for smaller sensor set sizes, but it steadily falls as $k$ increases.

\section{Proposed methodology}
\label{sec:methodology}
As mentioned before, the FOS approach exploits the centrality of the nodes of the graph in order to construct the sensor set in a epidemiological network. Mathematically, it can be explained by looking at the probability of being selected in the sensor set of a node. In an undirected graph with $n$ nodes with $d_j$ being the degree of the $j$-th node, this is given by the following \Cref{eq:FOS}.
    \begin{align}
    \nonumber
        &P(\text{node j is selected in sensor set})\\ \nonumber
        &= P(\text{at least one of the neighbours of node j is selected in the random sample})\\\nonumber
        &= 1 - P(\text{none of  the neighbours of node j are selected in the random sample})\\
        &= 1 - \frac{{{n-d_j} \choose k}}{{n \choose k}}
        \label{eq:FOS}
    \end{align}
This clearly shows that higher the degree is, higher is the probability of being selected in the sensor set. In this sense, this method utilizes the degree centrality of a network. However this centrality measure does not take into account the importance of the neighbors of an individual while determining its centrality. For example, a node with all of neighbors with degree $1$ would be assigned the same score in terms of the centrality as the one with same number of neighbors but some of them having degree more than $1$. To remedy this we propose a similar method that uses the eigenvector centrality. 

\subsection[EV Approach]{Eigenvector of the adjacency matrix (EV) approach}

The fundamental premise of the notion of the eigenvector centrality is, \textit{a node is important if it is neighbor to other important nodes.} This is in some sense an inductive concept. However, mathematically this can be expressed precisely. In this method relative centrality scores $\{v_1,\dots,v_n\}$ are assigned to all nodes in the network based on the concept that connections to high-scoring nodes should contribute more to the score of the node in question. This can be done recursively by initially taking $v_i^{(1)} = 1 \enskip \forall i=1,\dots,n$ and then defining a node importance at step $t+1$ as a function of the node importance of its neighbors at step $t$ as follows.
\begin{equation}
        v_i^{(t+1)} = \frac{1}{\lambda}\sum_{j}A_{ij}v_j^{(t)} \iff \lambda\textbf{v}^{(t+1)} = A\textbf{v}^{(t)},
        \label{eq:ev}
\end{equation} 
Here $\lambda$ is used for down-weighing the scores and facilitating convergence to the centrality scores. We also assume the network to be connected. Assuming the convergence of \Cref{eq:ev}, note that any eigenvector of the adjacency matrix could be the limiting value. However, the following theorem asserts the unique convergence of the given equation.
\begin{theorem}
\label{th:pf_theorem}
\textbf{(Perron-Frobenius Theorem)} Let ${\displaystyle A=(a_{ij})}$ be an ${\displaystyle n\times n}$ positive matrix, (i.e. $a_{{ij}}>0$ for ${\displaystyle 1\leq i,j\leq n}$). Then the following statements hold.
\begin{itemize}
    \item There is a positive real number $r$, known as the Perron-Frobenius eigenvalue , such that $r$ is an eigenvalue of $A$, and for any other eigenvalue $\lambda$ of $A$, its absolute value is strictly lesser than $r$. 
    \item Let $\vec{v}$ be the eigenvector corresponding to the eigenvalue $r$. Then all the components of $\vec{v}$ are positive, i.e. $v_i>0$, for $1\leq i\leq n$. Moreover there are no other positive eigenvectors of $A$ except for the positive multiples of $A$.
\end{itemize}
\end{theorem}
Note that obtaining eigenvector centrality scores include summing only nonnegative real numbers and hence score for a node can not be negative. Hence by \Cref{th:pf_theorem}, the only solution to which the above recurrence relation converges, is the largest eigenvector of the adjancency matrix of the graph $A$. Hence, we propose the following method leveraging the eigenvector centrality of the nodes of the underlying social network.
Select the first $k$ nodes based on the eigenvector centrality score (i.e. for node $i$, the $i$-th element of the eigenvector corresponding to the largest eigenvalue would be the score) in the sensor set. $k$ is the parameter of this method. Choice of $k$ is certainly is an important task, since choosing $k$ to be very high or very low might lead to the reduction in the lead time. However, in this write-up we would not focus on this. Instead, for the sake of comparison, the $k$ is taken as the size of the sensor set chosen by the FOS approach.
\subsection[NEV Appraoch]{Eigenvector of the column normalized adjacency matrix (NEV) approach}

This approach can be interpreted as a direct extension of the FOS approach. In FOS approach, more central nodes are selected by choosing the friends of a random sample. But one might wonder what would happen if this is repeatedly done. Would that lead to more central nodes? In this approach we have investigated the answer to this question. For the sake of simplicity start with selecting one node randomly from the graph. Then at each step we select one individual randomly from the friends of the individual selected in the previous step. This transition can be interpreted as a discrete time markov chain $\left\{X_i\right\}_{i\in \mathbb{N}}$ with the state space being the nodes of the network in consideration. We say that the chain in at state $i$ at time $t$ if the $i$-th node is selected at time $t$. Next consider the following lemma and the theorem.

\begin{lemma}
\label{lma:aperidocity}
Let's assume that the network in consideration has at least one odd cycle. Then the aforementioned markov chain is aperiodic and irreducible.
\end{lemma}

\begin{proof}
Irreducibility of the chain follows trivially from the fact that the underlying network is connected. To prove the irreducibility of the chain, note that since the network is undirected, it can be interpreted as a directed network with all the undirected edges being replaced by two directed edges (opposite direction). Which in turn means the presence of a cycle of length $2$. Now presence of a cycle of odd length would mean that the periodicity of the underlying graph is $1$, i.e. in parlance of graph periodicity, this network is graph-aperiodic. Now it follows from here that this underlying markov chain is also aperiodic.
\end{proof}

\begin{theorem}
 Under the condition of \Cref{lma:aperidocity}, the aforementioned markov chain has the limiting probability distribution given by the eigenvector corresponding to the eigenvalue $1$ of the matrix $B = AD$, where A is the adjacency  matrix of the network in consideration and $D = \text{diag}\{\frac{1}{|N(1)|},\dots,\frac{1}{|N(n)|}\}$.
\end{theorem}

\begin{proof}
let $\textbf{p}_t \in \mathbb{R}^{n \times 1}$ be the vector of probabilities of being selected in the sample at the $t^{th}$ step. It also denotes the probability vector corresponding to the chain at time $t$. Then,
    \begin{align*}
    \centering
        &\textbf{p}_1 = \frac{1}{n}\mathbf{1}_{n}\\
        &\textbf{p}_t|j^{th} \enskip \text{node was selected at time} \enskip (t-1) = \frac{1}{|N(j)|}{\mathbf{1}^{N(j)}_n}\\
        &\text{where } {\mathbf{1}^{N(j)}_{n,k}} = \begin{cases}
            1 & \text{ if node j and k are connected}\\
            0 & \text{otherwise}
        \end{cases}\\
        &\textbf{p}_t = \sum_{j=1}^{n}\frac{1}{|N(j)|}{\mathbf{1}^{N(j)}_n}p_{t-1,j} =  \begin{bmatrix}
            \textbf{v}_1 \dots \textbf{v}_n 
        \end{bmatrix}  
        \begin{bmatrix}
            p_{t-1,1}  \\
            \vdots  \\
            p_{t-1,n} 
        \end{bmatrix}  = B\textbf{p}_{t-1} \text{,say}\\
        &\text{where } B \text{ is the transition probability matrix of this markov chain, with } \textbf{v}_j = \frac{1}{|N(j)|}{\mathbf{1}^{N(j)}_n}
    \end{align*}
      Note that using \Cref{lma:aperidocity}, the limiting distribution of the chain is the stationary distribution, which is the eigenvector corresponding to the eigenvalue $1$ of the matrix $B$.
\end{proof}

\begin{remark}
By other extensions of the Perron-Frobenius theorem, it can be shown that $1$ is the Perron-Frobenius eigenvalue of $B$. Also, the fact that $1$ is indeed the eigen value of $B$ would come from the fact that eigenvalues of $B$ and $B^T$ are the same and it can be verified easily that $B^T\mathbf{1} = \mathbf{1}$.
\end{remark}
The NEV approach proposes to select the nodes based on the eigenvector corresponding to the eigenvalue 1 of B. Similar to EV approach, in this method also, we do not focus on the regime of determining the size of the sensor set. It is taken to be the size of the sensor set by FOS approach.

Below we provide the comparative summary of the methods we have discussed so far (\Cref{tab: method_comp}).
\begin{table}[!ht]
\centering
\begin{tabular}{@{}ll@{}}
\toprule
Method & Procedure \\ \midrule
FOS & \begin{tabular}[c]{@{}l@{}}Choose a random sample and \\ include their friends in the sensor set\end{tabular} \\\midrule
EV & \begin{tabular}[c]{@{}l@{}}select the nodes based on the updated \\ scores stored in the eigenvector corresponding \\ to the largest eigenvalue of A.\end{tabular} \\\midrule
NEV & \begin{tabular}[c]{@{}l@{}}select the nodes based on the eigenvector \\ corresponding to the eigenvalue 1 of B.\end{tabular} \\ \bottomrule
\end{tabular}
\label{tab: method_comp}
\caption{Brief summary of the methods discussed above.}
\end{table}

\subsection{Estimation of parameters and the lead time}

The previous heuristic methods did not leverage any information on the disease propagation model itself. So our next approach would be to use that information in order to predict the lead time. One simple but computationally expensive way could be to run an SIR simulation based on estimated $\beta$ and $\gamma$ values to analytically determine the peak time and hence the lead time. The first and foremost step of that would be to estimate the parameters used to define the disease SIR model $\beta$ and $\gamma$. The Maximum Likelihood Estimators of these quantities cannot be derived analytically as the likelihood is very complex. So we provide simple Method of Moments type unbiased estimators of $\beta$ and $\gamma$. We first define $$I_{it} = \begin{cases}
            1 & \text{ if node i is infected at time t}\\
            0 & \text{otherwise}
        \end{cases}$$
$$S_{it} = \begin{cases}
            1 & \text{ if node i is susceptible at time t}\\
            0 & \text{otherwise}
        \end{cases}$$
$$R_{it} = \begin{cases}
            1 & \text{ if node i is recovered at time t}\\
            0 & \text{otherwise}
        \end{cases}$$
Then, define $\hat{\beta} = \frac{1}{nT}\sum_{t=2}^{T} \sum_{i=1}^{n} \frac{I_{it} S_{i,t-1}}{\sum_{j \in N(i)}I_{j,t-1}}$ and $\hat{\gamma} = \frac{1}{nT}\sum_{t=2}^{T} \sum_{i=1}^{n} {R_{it} I_{i,t-1}}$. To show the unbiasedness we use the smoothing formula of expectation,
\begin{align*}
\allowdisplaybreaks
    E(\hat{\beta}) & = \frac{1}{nT}\sum_{t=2}^{T} \sum_{i=1}^{n} E \left[ \frac{I_{it} S_{i,t-1}}{\sum_{j \in N(i)}I_{j,t-1}} \right] \\
    & = \frac{1}{nT}\sum_{t=2}^{T} \sum_{i=1}^{n} E\left[\frac{S_{i,t-1}}{\sum_{j \in N(i)}I_{j,t-1}} E(I_{it}| \mathcal{A}_{t-1})\right] \end{align*}
    where $\mathcal{A}_{t-1}$ is the sigma field generated by $\{S_{i,t-1}\}$ and $\{I_{j,t-1}\}$.
    
\begin{align*}
    \qquad \qquad & = \frac{1}{nT}\sum_{t=2}^{T} \sum_{i=1}^{n} E\left[\frac{S_{i,t-1}}{\sum_{j \in N(i)}I_{j,t-1}} P(I_{it}=1|\mathcal{A}_{t-1}) \right] \\
     \qquad \qquad & = E\left[\frac{1}{\sum_{j \in N(i)}I_{j,t-1}} P(I_{it} = 1| S_{i,t-1} = 1)\right]\\
     \qquad \qquad & = \beta
\end{align*}
And similarly, defining a sigma algebra over $I_{i,t-1}$, we can write
\begin{align*}
    E(\hat{\gamma}) & = \frac{1}{nT}\sum_{t=2}^{T} \sum_{i=1}^{n} E \left[ R_{it} I_{i,t-1} \right]\\
    & = \frac{1}{nT}\sum_{t=2}^{T} \sum_{i=1}^{n} E\left[I_{i,t-1} P(R_it = 1|I_{i,t-1} = 1)\right]\\
    & = \gamma
\end{align*}
Note that here $T$ is not the entire period of the disease propagation, rather it is the time till when we have observed the propagation, very likely a few time points after the lead time in the sensor group. However, simulations studies suggest that these estimators largely under-estimate the actual parameter values, hence we have not proceeded much further in this direction. In future work, we plan to further investigate the reasons for this, and then we would like to look at the variance of these estimators. The closed form expression could be derived from the smoothing formula on the variance operator.

\section{Simulation study}
\label{sec:sim}
\subsection{Social network models for simulation}

The degree of a node in an undirected graph is the number of edges it has. The degree distribution is the probability mass function for all the degrees, i.e., the distribution of degree we would find by picking randomly and uniformly over nodes. For our purpose, we consider two broad groups of random graphs; egalitarian or decentralized (one where the edges are more or less equally distributed across the network, that is a more or less uniform degree distribution) and authoritarian or centralized (where a central node has a much higher degree than non-central nodes) \citep{sueur2012social}. We choose the Erd\"{o}s-R\'{e}nyi and Chung-Lu models to represent the two kinds of graphs respectively. 

Erd\"{o}s-R\'{e}nyi networks are random graph with $n$ vertices where each possible edge has probability $p$ of existing. Consider a graph with $n$ nodes. The full graph will then consist of $N = {n \choose 2}$ edges, and say the set of all possible edges is $E = \{e_1, e_2, \dots, e_N\}$. In an Erdos Renyi random graph $G(n,p)$, an edge $e_i$ is picked with probability $p$ independently of occurrences of other edges. Let $X$ be the number of edges in the graph. Then, The expected number of edges in this graph is ${n \choose 2}p$.
The expected mean degree in such a network is $(n-1)p$, which is the same for all vertices, implying that it pertains to our condition of being a representative of the decentralized society.


Chung-Lu Networks are random graphs  with n vertices where each possible edge has probability $p_{ij} = \frac{w_iw_j}{\sum_{k}w_k}$ of existing between nodes $i$ and $j$, where $\{w_1,\dots,w_n\}$ is a set of weights attached the $n$ nodes. Here, for a pair (i,j) an edge $e_i$ is chosen independently with probability $p_{ij}$. The expected degree of vertex $i$ is: 
    \begin{align*}
        \sum_{j=0}^{n}\frac{w_iw_j}{\sum_{k}w_k} = w_i
    \end{align*}
The edge distribution here thus depends upon the centrality of the nodes in this graph. The weights can be modified to create star-like topology which are vital to our study.


\subsection{Disease Propagation Model}
\label{sec:disease_model}
 We used the simple network based SIR model for our simulation study. Let $G_s = (V, E_s)$ be a social network. Based on this, the disease network would progress over time as follows:

\begin{figure}[!ht]
    \centering
    \includegraphics[width = 0.75
        \textwidth]{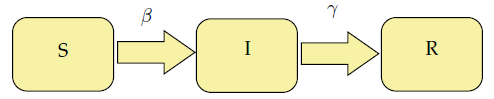}
\end{figure}

Consider the following probabilities,
\begin{align*}
        &S_i(t) = P(\text{node i is susceptible at time t})\\
        &I_i(t) = P(\text{node i is infected at time t})\\
        &R_i(t) = P(\text{node i has recovered at time t})
    \end{align*}
    
Clearly, at any given time $t$, $S_i(t)+I_i(t)+R_i(t) = 1$ should hold for all vertices.
The disease propagation on the underlying network structure can then be modelled as
\begin{align*}
    &\frac{dI_i}{dt} = \beta S_i \sum_{j}A_{ij}I_j - \gamma I_i\\
    &\frac{dS_i}{dt} = -\beta S_i \sum_{j}A_{ij}I_j, \quad \frac{dR_i}{dt} = \gamma I_i
\end{align*}
An approximate solution to the above set of differential equations is as follows:
    \begin{equation*}
        I(t) \approx \textbf{v}_1 e^{(\beta \lambda_1 - \gamma)t}
    \end{equation*}
    where $\lambda_1$ is the largest eigenvalue of the adjacency matrix A and $\textbf{v}_1$ is the corresponding eigenvector.
    
    For the ease of simulation, we consider discrete time setup with the probabilities of transition of node i from the state at time t (say $x_{i,t}$) to the state at time t+1 (say $x_{i,t+1}$) as follows:
    
    \begin{align*}
        (x_{i,t+1}\,|\,x_{i,t} = S) & = 
    \begin{cases}
      I, & \text{w.p. } \beta{\sum_{j \in \mathcal{N}(i)}\mathbb{I}(x_{jt} = I)}, \\
      S, & \text{w.p. } 1-\beta {\sum_{j \in \mathcal{N}(i)}\mathbb{I}(x_{jt} = I)}.
    \end{cases}\\
    (x_{i,t+1}\,|\,x_{i,t} = I) & = 
    \begin{cases}
      R, & \text{w.p. } \gamma, \\
      I, & \text{w.p.} 1-\gamma.
    \end{cases}\\
    (x_{i,t+1}\,|\,x_{i,t} = R) & = 
    \begin{cases}
      R, & \text{w.p. } 1.
    \end{cases}
    \end{align*}
    
\subsection{Simulation set-up}

\subsubsection{Algorithm for simulation}

We run a simulation study to compare among the performances of the three approaches discussed above to select sensors who would give a lead time before an epidemic peaks in the population on the whole. Our sensor group selection was mainly driven by finding and choosing nodes that are more and more central to the graph. To see this, we generated a social contact network with $n=1000$ individuals in the population. Each individual is meant to represent a node in the network and an edge between any two nodes represent mutual contact between the two individuals. We generated the Erd\"{o}s-R\'{e}nyi Model (ER) with probability of an edge $p = 0.005$, which leads to $m \approx 2500$ edges spread more or less uniformly across the entire network. We also generated the Chung Lu Network (CL) with $n = 1000$ and $m = 2500$, to keep parity with the ER Graph. However, the network structure here is star-like. We considered $T = 10000$. Every node at any time step is in one of the states $\{S, I, R\}$. We then initialized 10 nodes in state $I$ at $T = 1$. On every time step, each $I$ node has a probability $\beta = 0.001$ to infect its neighbours ($S$ $\rightarrow$ $I$) and $\gamma = 0.001$ to recover ($I$ $\rightarrow$ $R$). Once the disease propagated through the entire network over $T = 10000$ time-points following the defined SIR model, we selected a random sample of $10$ individuals from the entire network and select their friends (neighbours) in the FOS sensor group. For the EV and NEV approaches, the sensor group size was determined by the number of neighbours included in the FOS sensor group. The same was repeated for FOS, EV and NEV approaches beginning with $20$ random samples from the network.

\subsubsection{Parameters of the simulation study}
Here, we study how the peak times differ if we vary $\beta$ and $\gamma$ values. The following three cases are possible. Firstly, $\beta > \gamma$ where the rate at which the infection spreads is faster than the recovery rate. We take $\beta = 0.005$ and $\gamma = 0.002$. Next, we can have $\beta = \gamma$. Although, we already demonstrated the results for this case, we would like to see the differences, if any, when the rates are increased, so that the disease propagation is faster than before. A higher value of $\beta$ would ensure that the infection spreads quickly so it achieves the peak quite early and a larger $\gamma$ means that they recover quickly too. For this we choose $\beta = \gamma = 0.005$. And lastly, we can have the case when infection rate is slower than recovery, that is, $\beta < \gamma$, for which we take $\beta = 0.002$ and $\gamma = 0.005$.

\subsubsection{Estimation of Time to Peak}
We demonstrate the estimation of the population peak time of infection by the EV and NEV approaches for Erd\"{o}s-R\'{e}nyi Model (ER) and the Chung Lu Network (CL) with initially randomly assigning the state of infection to 10 individuals. We do this by regressing the cumulative infection per unit time for the EV (or NEV) sensor groups that are formed based on the number of neighbours of 20 randomly selected nodes to the population cumulative infection per unit time. A cubic degree polynomial is fitted and then the peak time for the population is predicted as per this regression. We use the data from the first few days (upto 100 time points after the peak in the sensor group) to estimate our polynomial regression model and make predictions about the cumulative incidence of the population for the first time point to the next 500 units of time after the sensor group peaked.

\subsection{Results}

\begin{figure}[!ht]
        \centering
        \includegraphics[width = 0.9
        \textwidth]{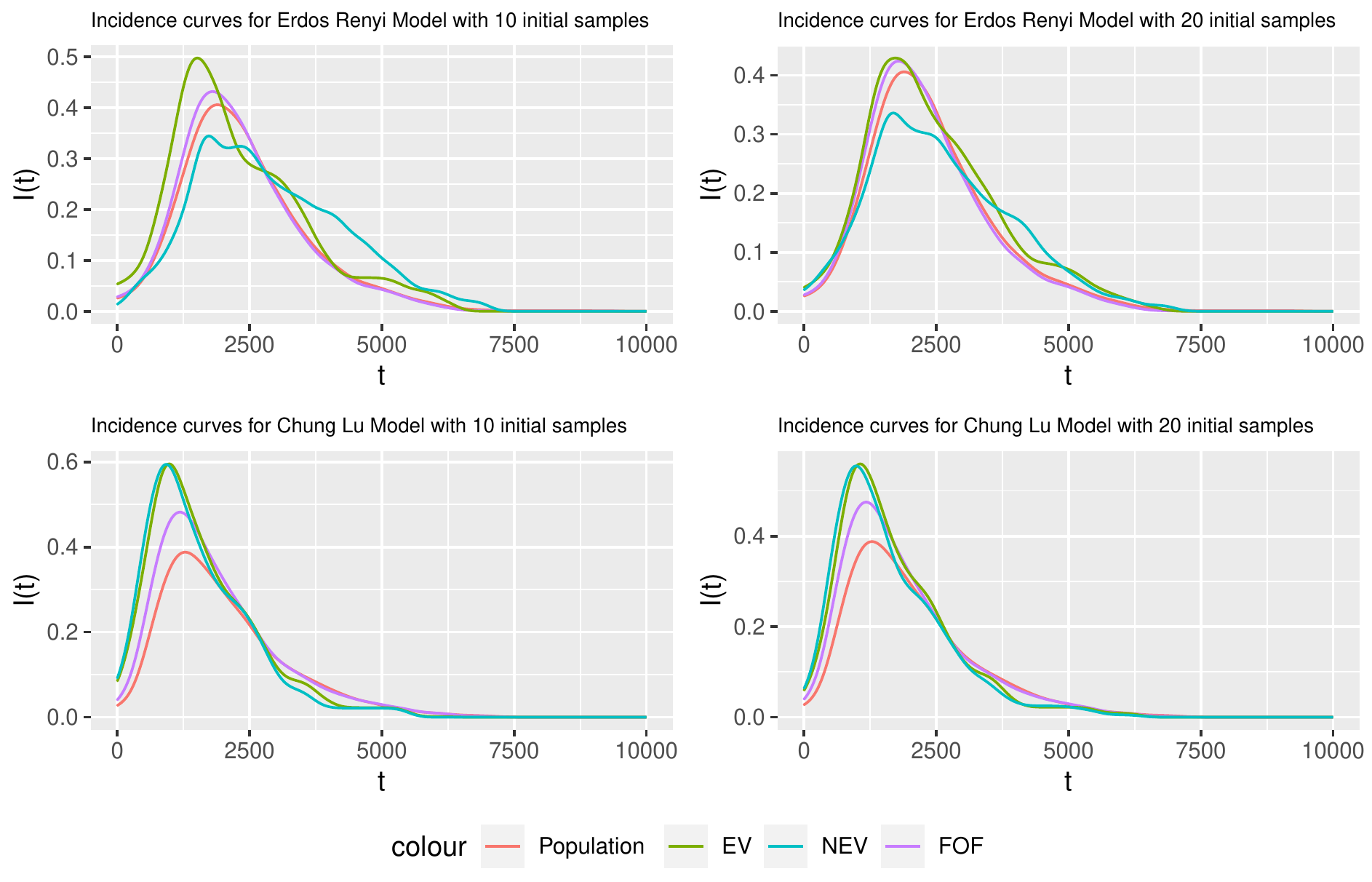}
        \caption{Plot of Incidence Curves for the ER Model and CL Model with 10 initial random samples (\textit{left}) and 20 initial random samples (\textit{right})}
        \label{fig:curves}
    \end{figure}

From the plots in {Figure} \ref{fig:curves} and also the results in \cref{tab:my-table}, we can see that EV and NEV approaches to sensor group selection give a lead time ahead of the FOS approach , which in turn peaks before the entire network on the whole. \cite{shao2016forecasting} had  noted that networks with star-like topology where a few of the central nodes have very large degrees, perform relatively better under the FOS approach as this graph structure facilitates inclusion of central nodes with high degrees to form the sensor group. However, our proposed methods give greater lead times under such star-like graphs generated by the Chung-Lu model. Moreover, we also noted that the lead time difference between the FOS and EV/NEV methods are more pronounced when smaller random samples are chosen to begin with. Having said that, henceforth we provide results for the cases where we begin with larger (20) initial samples, hoping that the corresponding outcomes would be better if smaller (10) samples were initially chosen. 

\begin{table}[!ht]
\centering
\caption{Peak Times of the three approaches compared to that of the whole population for the Erdos Renyi and Chung Lu Graph structures with initial sample sizes 10 and 20}
\label{tab:my-table}
\begin{tabular}{@{}ccccc@{}}
\toprule
\multirow{2}{*}{\begin{tabular}[c]{@{}c@{}}Network \\ Model\end{tabular}} & \multicolumn{4}{c}{Peak Time} \\ \cmidrule(l){2-5} 
 & Population & FOS & EV & NEV \\ \midrule
ER 10 & 1892 & 1801 & 1514 & 1730 \\
ER 20 & 1892 & 1787 & 1726 & 1687 \\
CL 10 & 1284 & 1185 & 983 & 925 \\
CL 20 & 1284 & 1174 & 1062 & 985 \\ \bottomrule
\end{tabular}
\end{table}

Next, for the different values of infection and recovery rates, we present the peak times for the four groups in \ref{tab:diff_rates}. Our proposed methods not only work better than the exsiting FOS approach under the star-like Chung-Lu Model under all the three variations in the values of thr rate parameters, but also perform comparably under the Erd\"{o}s-R\'{e}nyi Network. 

\begin{table}[!ht]
\centering
\caption{Peak Times of the three approaches compared to that of the whole population for the Erdos Renyi and Chung Lu Graph structures with different rate parameters.}
\label{tab:diff_rates}
\begin{tabular}{@{}lcccc@{}}
\toprule
\multirow{2}{*}{\begin{tabular}[c]{@{}c@{}}Network \\ Model\end{tabular}} & \multicolumn{4}{c}{Peak Time} \\ \cmidrule(l){2-5} 
 & Population & FOS & EV & NEV \\ \midrule
ER 20 ($\beta = 0.005, \gamma = 0.002$) & 446 & 426 & 445 & 446 \\
ER 20 ($\beta = \gamma = 0.005$) & 366 & 354 & 357 & 357 \\
ER 20 ($\beta = 0.002, \gamma = 0.005$) & 1235 & 1213 & 894 & 1106 \\
\midrule
CL 20 ($\beta = 0.005, \gamma = 0.002$)
 & 346 & 292 & 240 & 217 \\ 
CL 20 ($\beta = \gamma = 0.005$)
 & 358 & 320 & 279 & 269 \\
CL 20 ($\beta = 0.002, \gamma = 0.005$)
 & 511 & 443 & 412 & 367 \\\bottomrule
\end{tabular}
\end{table}

Finally, we report the estimated peaks from the sensor groups. The polynomial regression of degree three fits the data quite well and {Figure} \ref{fig:est1} shows that the estimated peak time by the EV sensor group for the Erd\"{o}s-R\'{e}nyi Model is 2000, whereas the population actually peaks on the time point 1892. However, for the Chung Lu Model, the EV sensor group estimates the peak time to be 1144, whereas the population actually peaks on 1284. The absolute peak time difference for the ER model is 108 and that for the CL model is 140. We note here that the simulation was conducted over 10000 time points and so estimation of  the margin of error of the peak time is acceptable here. 

\begin{figure}[!ht]
        \centering
        \includegraphics[width = 0.9
        \textwidth]{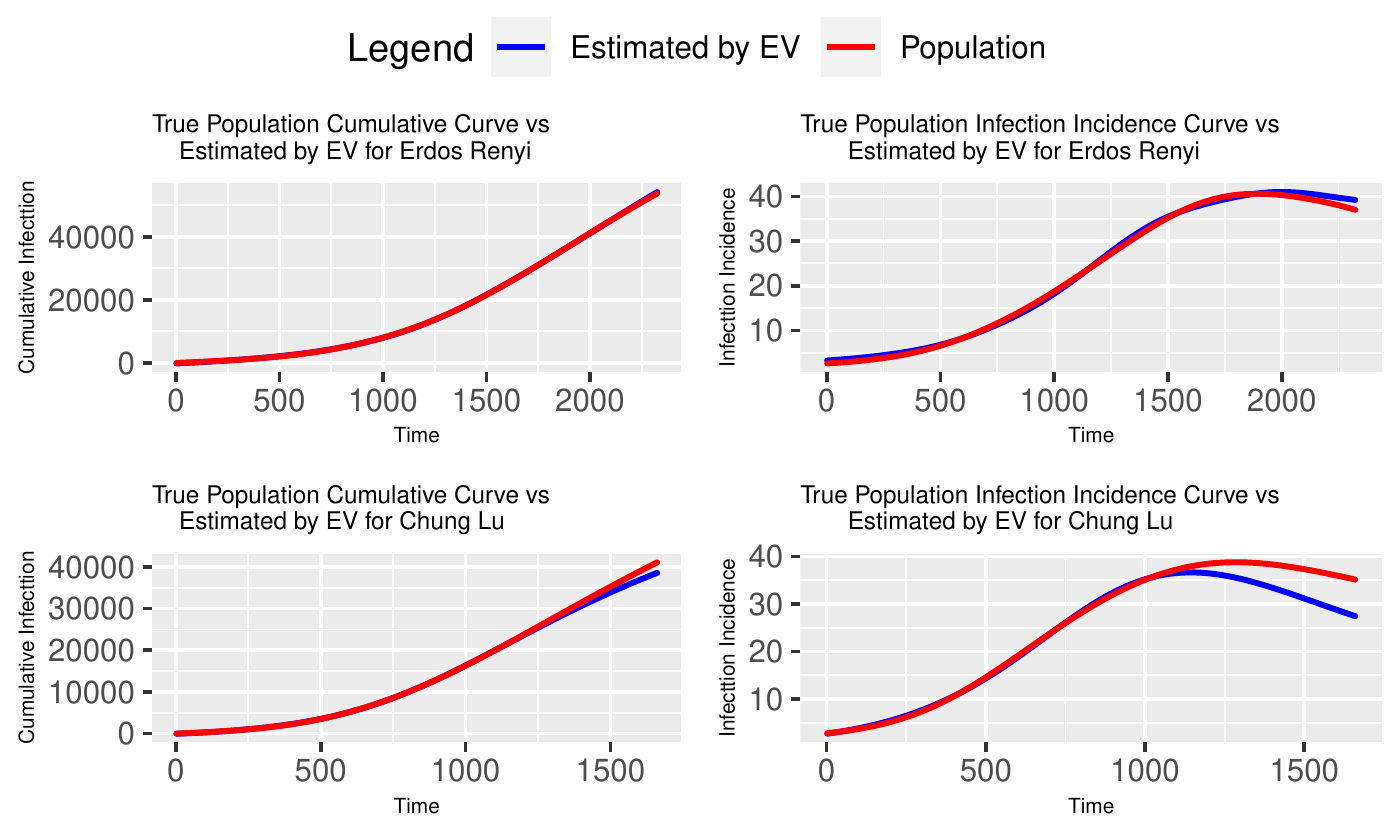}
        \caption{Plot of Incidence Curves for the ER Model (\textit{above}) and CL Model (\textit{below}) comparing the peak time of the population and the peak time estimated by the EV sensor group. }
        \label{fig:est1}
    \end{figure}

For the NEV sensor group with initially 20 infected individuals, we notice from {Figure} \ref{fig:est2} that for the Erd\"{o}s-R\'{e}nyi Model, the estimated peak time is 1588, which is 304 time points ahead of the actual time to peak by the population 1892. Again, for the Chung Lu Network, the estimated peak time is 1259, whereas the population peak time is 1284, resulting in a margin of error of 25. The NEV approach is seen to work better for star-like graph structures and that explains why the error in estimation under Chung Lu Model is lower. 

\begin{table}[!ht]
\centering
\label{tab:bias}
\begin{tabular}{@{}cccccc@{}}
\toprule
\multirow{2}{*}{\begin{tabular}[c]{@{}c@{}}Network \\ Model\end{tabular}} & \multicolumn{5}{c}{Peak Time} \\ \cmidrule(l){2-6} 
 & Population & \multicolumn{1}{p{2cm}}{\centering Estimated \\  by EV} &  \multicolumn{1}{p{2cm}}{\centering Bias due \\  to EV} & \multicolumn{1}{p{2cm}}{\centering Estimated \\  by NEV} &  \multicolumn{1}{p{2cm}}{\centering Bias due \\  to NEV} \\ \midrule
ER 20 & 1892 & 2000 & 108 & 1588 & 304 \\
CL 20 & 1284 & 1144 & 140 & 1259 & 25 \\ \bottomrule
\end{tabular}

\caption{Estimated Peak Times of the two approaches compared to that of the whole population for the Erdos Renyi and Chung Lu Graph structures with initial infected samples of 20.}
\end{table}

\begin{figure}[!ht]
        \centering
        \includegraphics[width = 0.9
        \textwidth]{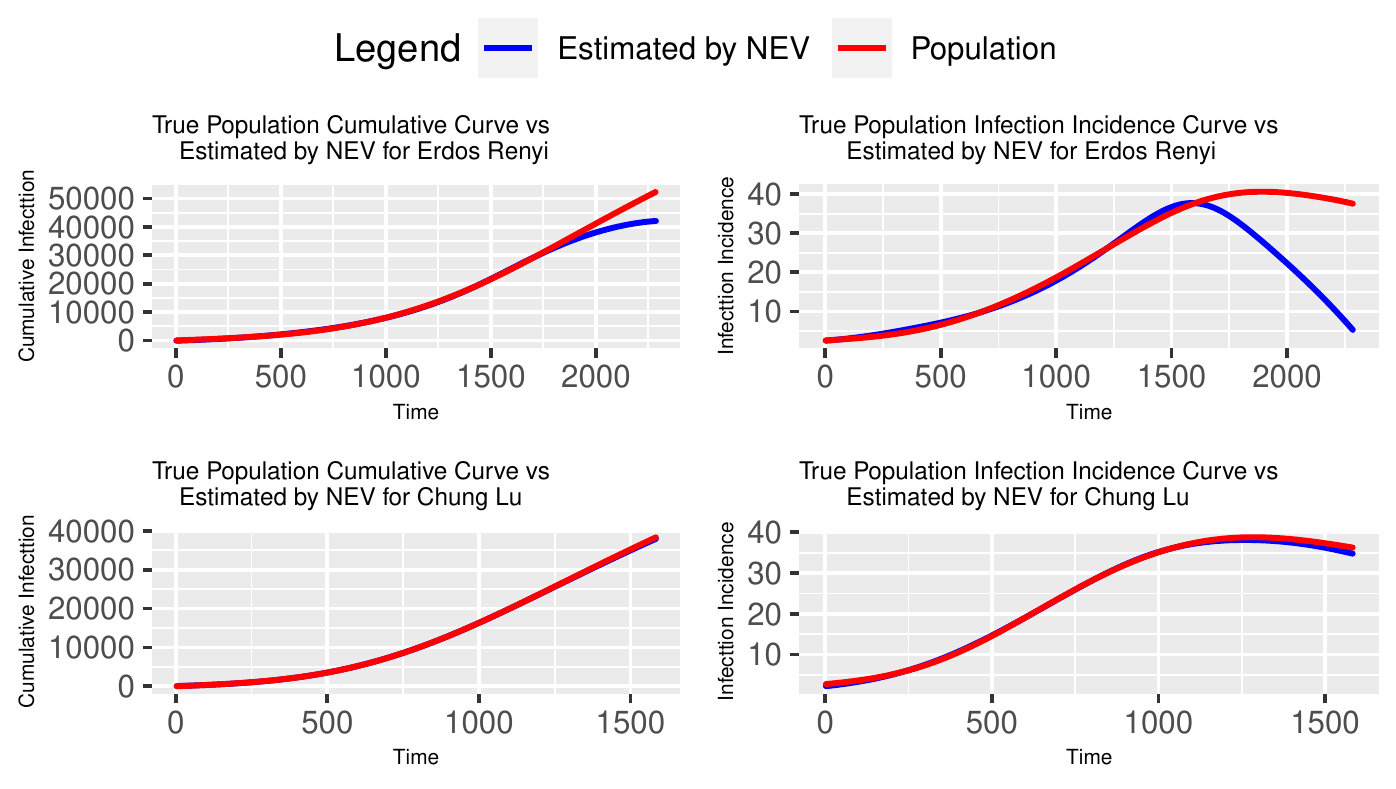}
        \caption{Plot of Incidence Curves for the ER Model (\textit{above}) and CL Model (\textit{below}) comparing the peak time of the population and the peak time estimated by the NEV sensor group. }
        \label{fig:est2}
    \end{figure}

\section[Real Data Analysis]{Sensor set selection based on contact patterns in a village in rural Malawi}
\label{sec:realdata}


 \begin{figure}[!ht]
        \centering
        \includegraphics[width = 0.85
        \textwidth]{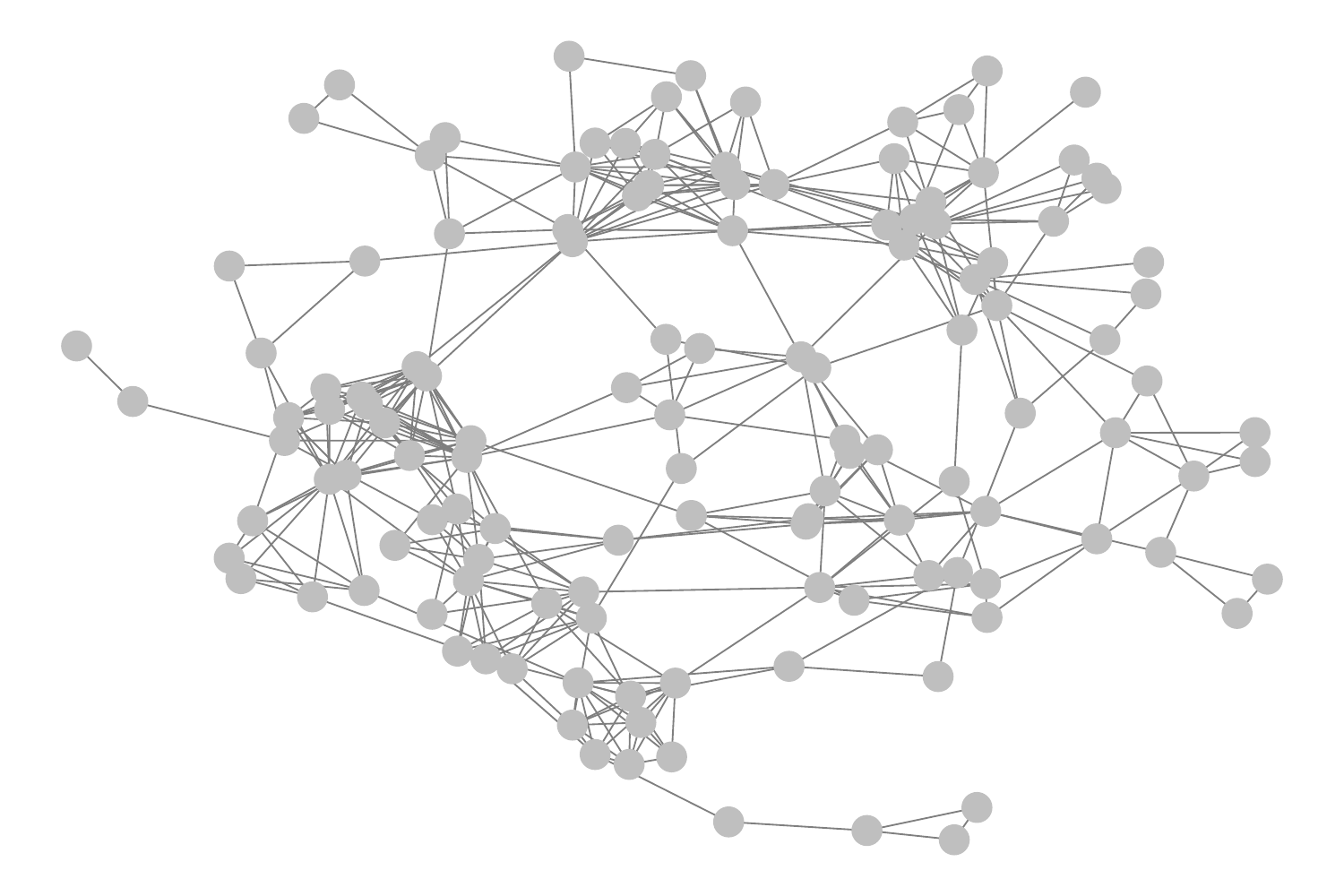}
        \caption{Connected part of the contact network of rural Malawi.}
        \label{fig:real}
    \end{figure}


\begin{figure}
\centering
  \includegraphics[width = 0.45 \textwidth]{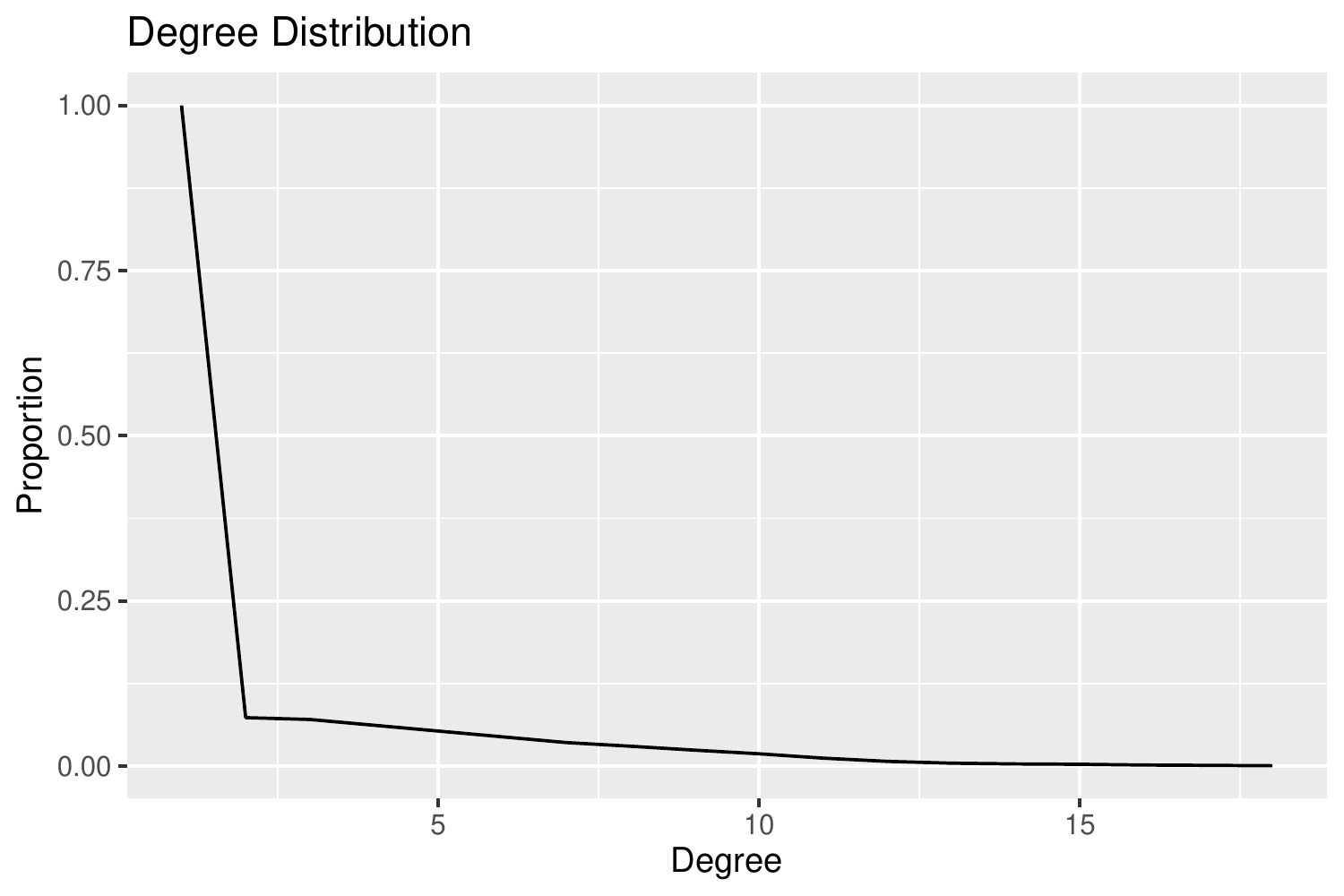}
  \includegraphics[width = 0.45
        \textwidth]{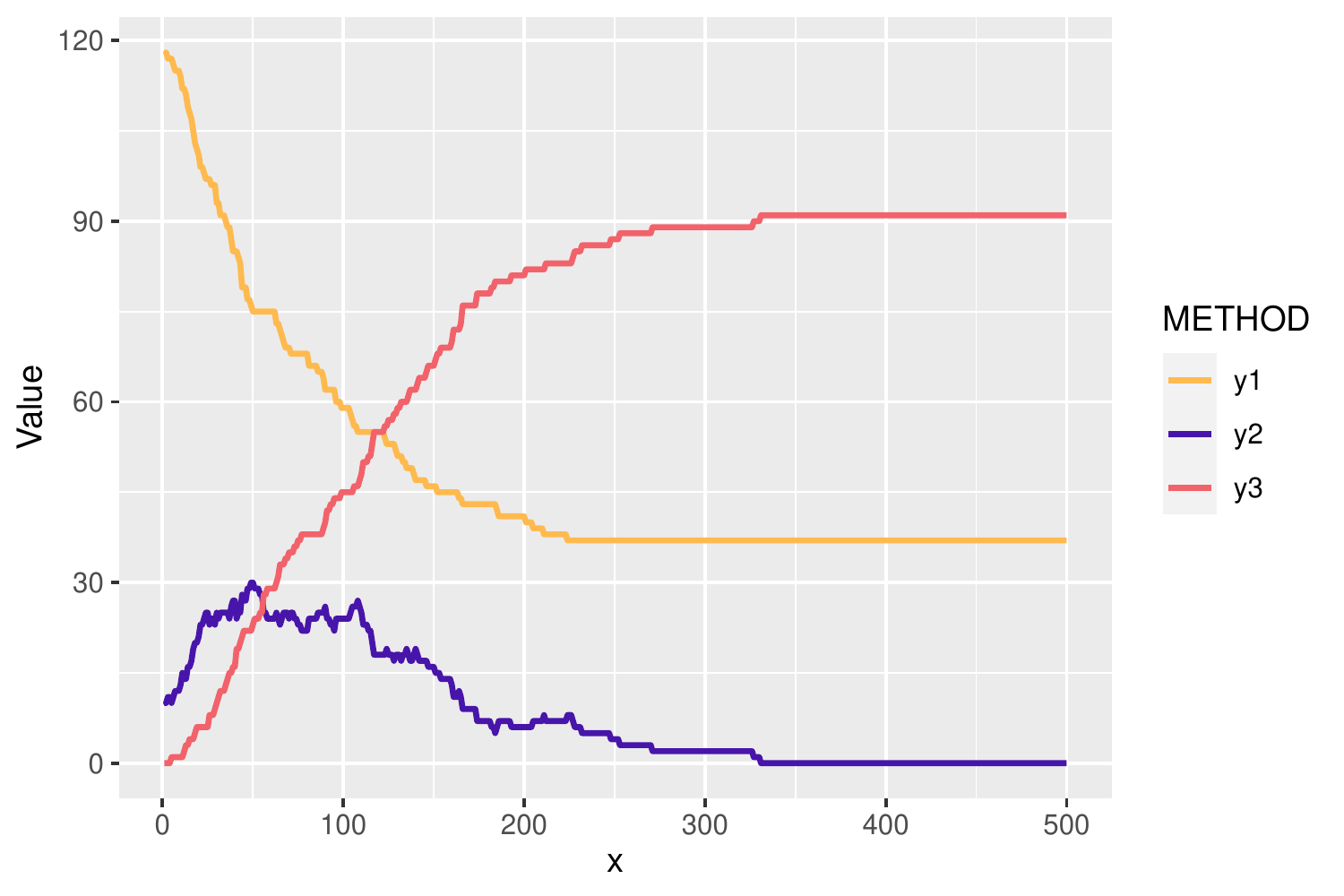}
        \caption{Left: Degree Distribution of the entire contact network of rural Malawi. Right: Incidence Curves for the connected part of the contact network of rural Malawi following the simulated SIR disease propagation model.\label{fig:inci}}
        
\label{fig:test}
\end{figure}

We now report results from a real-world social contact network recently curated by \cite{ozella2021using}.
The data was collected in Mdoliro village, which is a small settlement (147 people living in 32 homes in 2019) in Dowa district, Malawi's Central Region,
between the 16th of December 2019 and the 10th of January 2020.
Two people are defined to have a 'contact event'  when the proximity-sensing devices exchanged at least one radio packet in a 20-second period. However, for the EV and NEV methods, we only consider the connected component of this network, that includes 128 nodes and 401 edges. The fraction of nodes in a network of degree $k$ is defined as the network's degree distribution $P(k)$. From the left panel of {Figure} \ref{fig:inci}, we see that the degree distribution follows the power law. This is indicative of a few nodes being highly connected to other nodes in the network. The degree distribution $P(k)$ decays slowly as the degree $k$ increases, increasing the likelihood of finding a node with a very large degree.


We now propagate our simulated SIR model through the connected portion of this network using $\beta = 0.025$ and $\gamma = 0.045$ for 500 time points. We begin by randomly choosing 10 nodes that are infected at the first time point. The Incidence curves of infectious, susceptible and recovered individuals over time is for this network is given in the right panel of {Figure} \ref{fig:inci}. Next, we select 20 individuals at random and select their immediate neighbours to form the FOS sensor group. Based on this FOS sensor group size, the EV and NEV sensor groups are selected and the disease model is propagated through the four different sets of individuals and their infection curve is displayed in {Figure} \ref{fig:real2}. The EV and NEV sensor groups reach the height of infection before the population, and more importantly, before the FOS sensor group as well. The times to peak and lead times are tabulated below in \ref{tab:tab_real}. The NEV sensor group peaks much ahead of the population and the EV sensor group still gives some lead time, but the FOS sensor group peaks after the peak time of the population. 
\begin{table}[]
\centering
\caption{Lead Times by the 3 sensor groups}
\label{tab:tab_real}
\begin{tabular}{@{}lcc@{}}
\toprule
{ \textbf{Group}} & {\textbf{Peak Time}} & {\textbf{Lead Time}} \\ \midrule
Population & 62 & -  \\
FOS Sensor & 65 & -3 \\
EV Sensor  & 56 & 6  \\
NEV Sensor & 41 & 21 \\
\bottomrule
\end{tabular}
\end{table}
 \begin{figure}[!ht]
        \centering
        \includegraphics[width = 0.8
        \textwidth]{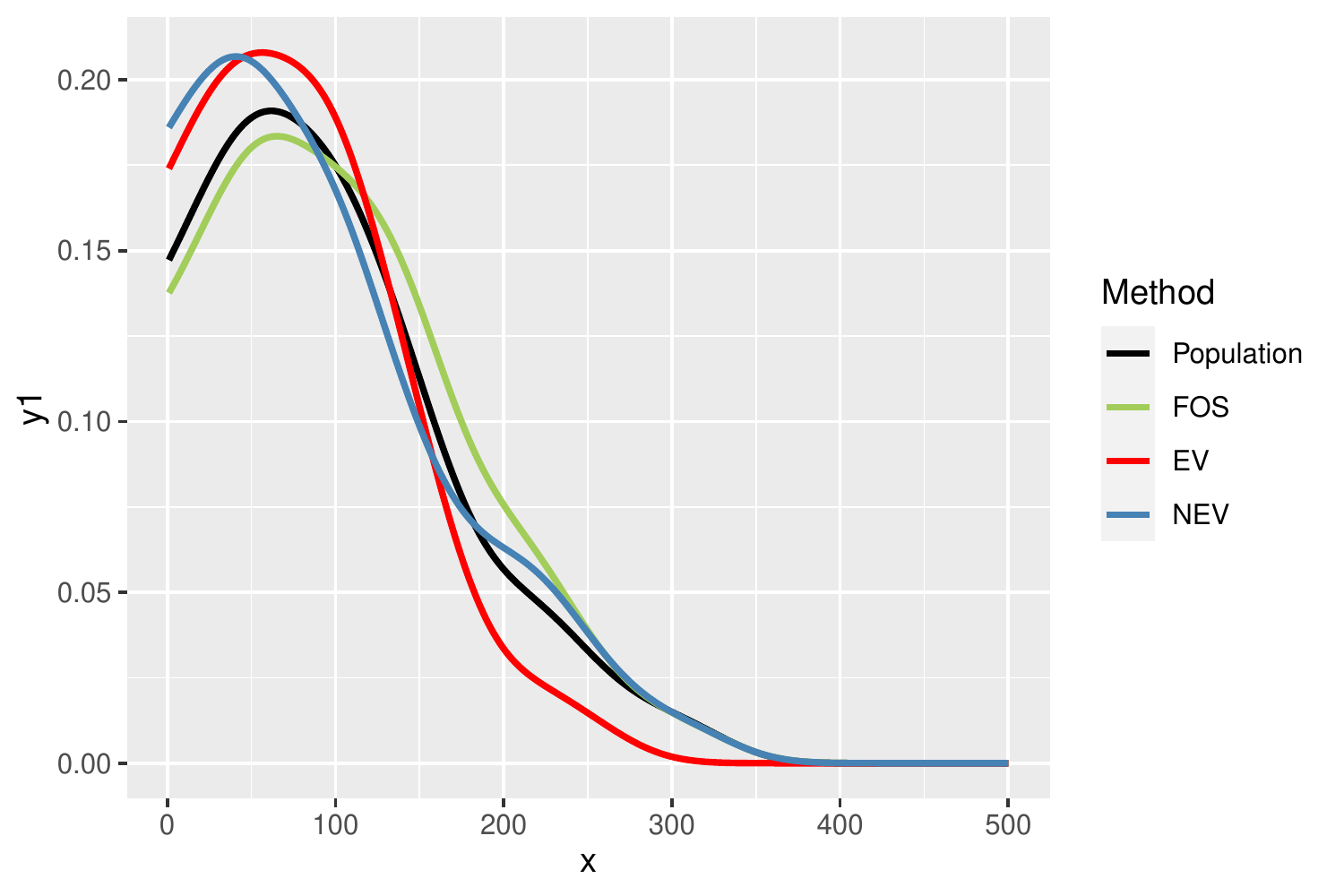}
        \caption{Infection Curves comparing Population with FOS, EV and NEV sensor groups.}
        \label{fig:real2}
    \end{figure}

\section{Conclusion and future directions}
\label{sec:disc}

In this write-up we have presented our findings based on the simulation studies, the goal of which was to compare the existing method (FOS) by \cite{fowler} and the ones proposed by us. The results from the simulation studies suggests that our proposed methods perform better than the FOS method. Although we do not have any rigorous theoretical justification in this regard, we may provide some intuitive explanation for this. Firstly, the EV centrality is a more robust approach of centrality, since it takes into account more information, which is the importance of the friends of a node, while determining its centrality score. Moreover, the importance of EV centrality notion is also reflected in the approximate solution of the incidence curve under network SIR model as shown in \Cref{sec:disease_model}. On the other hand, the NEV approach could be considered to be an extension of the method used in FOS approach as it was discussed previously. This intuitively suggests that the NEV approach suggests towards more central nodes, and in turn improves the peak time. We extended our simulation study to several choices of the rate parameters used in the disease generation model and noticed that our proposed methods work much better than FOS under star-like network structures. We also used the cumulative infections from the EV and NEV sensor groups to estimate the infection curve of the population for both ER and CL network models using a cubic polynomial regression. This would definitely help in estimating the population time to peak and consequently give us an estimate of the lead time that each of the methods provide. Finally, we analyse a real contact-network data from a village in rural Malawi and show that NEV works the best. Some future directions could be to extend our methods to the scenario where the disease progression depends on random-mixing along with the underlying social network structure. Also, note that in order to make our proposed methods objective, one should come up with a method to decide the optimal cardinality of the sensor set, which would require an extensive simulation study to get an idea on how the lead time changes with the cardinality of the sensor set.


\bibliography{references,ref}

\begin{thebibliography}{38}
\providecommand{\natexlab}[1]{#1}
\providecommand{\url}[1]{\texttt{#1}}
\expandafter\ifx\csname urlstyle\endcsname\relax
  \providecommand{\doi}[1]{doi: #1}\else
  \providecommand{\doi}{doi: \begingroup \urlstyle{rm}\Url}\fi

\bibitem[Adamic and Glance(2005)]{adamic2005political}
Lada~A Adamic and Natalie Glance.
\newblock The political blogosphere and the 2004 {US} election: divided they
  blog.
\newblock In \emph{Proceedings of the 3rd International Workshop on Link
  Discovery}, pages 36--43. ACM, 2005.

\bibitem[Albert and Barab{\'a}si(2002)]{albert2002statistical}
R{\'e}ka Albert and Albert-L{\'a}szl{\'o} Barab{\'a}si.
\newblock Statistical mechanics of complex networks.
\newblock \emph{Reviews of modern physics}, 74\penalty0 (1):\penalty0 47, 2002.

\bibitem[Barab{\'a}si and Albert(1999)]{barabasi1999emergence}
Albert-L{\'a}szl{\'o} Barab{\'a}si and R{\'e}ka Albert.
\newblock Emergence of scaling in random networks.
\newblock \emph{Science}, 286\penalty0 (5439):\penalty0 509--512, 1999.

\bibitem[Bengtsson et~al.(2015)Bengtsson, Gaudart, Lu, Moore, Wetter, Sallah,
  Rebaudet, and Piarroux]{bengtsson2015using}
Linus Bengtsson, Jean Gaudart, Xin Lu, Sandra Moore, Erik Wetter, Kankoe
  Sallah, Stanislas Rebaudet, and Renaud Piarroux.
\newblock Using mobile phone data to predict the spatial spread of cholera.
\newblock \emph{Scientific reports}, 5:\penalty0 8923, 2015.

\bibitem[Bickel and Chen(2009)]{bickel2009nonparametric}
Peter~J. Bickel and A.~Chen.
\newblock A nonparametric view of network models and {N}ewman--{G}irvan and
  other modularities.
\newblock \emph{Proceedings of the National Academy of Sciences}, 106:\penalty0
  21068--21073, 2009.

\bibitem[Bickel and Sarkar(2016)]{bickel2015hypothesis}
Peter~J Bickel and Purnamrita Sarkar.
\newblock Hypothesis testing for automated community detection in networks.
\newblock \emph{Journal of the Royal Statistical Society: Series B (Statistical
  Methodology)}, 78\penalty0 (1):\penalty0 253--273, 2016.

\bibitem[Bogu{\~{n}}{\'{a}} and
  Pastor-Satorras(2002)]{Boguna2002EpidemicNetworks}
Marián Bogu{\~{n}}{\'{a}} and Romualdo Pastor-Satorras.
\newblock {Epidemic spreading in correlated complex networks}.
\newblock \emph{Physical Review E}, 66\penalty0 (4):\penalty0 047104, 2002.
\newblock ISSN 1063-651X.
\newblock \doi{10.1103/PhysRevE.66.047104}.

\bibitem[Castellano and Pastor-Satorras(2010)]{Castellano2010}
Claudio Castellano and Romualdo Pastor-Satorras.
\newblock {Thresholds for Epidemic Spreading in Networks}.
\newblock \emph{Physical Review Letters}, 105\penalty0 (21):\penalty0 218701,
  2010.
\newblock ISSN 0031-9007.
\newblock \doi{10.1103/PhysRevLett.105.218701}.

\bibitem[Chakrabarti et~al.(2008)Chakrabarti, Wang, Wang, Leskovec, and
  Faloutsos]{Chakrabarti2008}
Deepayan Chakrabarti, Yang Wang, Chenxi Wang, Jurij Leskovec, and Christos
  Faloutsos.
\newblock {Epidemic thresholds in real networks}.
\newblock \emph{ACM Transactions on Information and System Security},
  10\penalty0 (4):\penalty0 1--26, 2008.
\newblock ISSN 10949224.
\newblock \doi{10.1145/1284680.1284681}.

\bibitem[Christakis and Fowler(2010)]{fowler}
Nicholas~A. Christakis and James~H. Fowler.
\newblock Social network sensors for early detection of contagious outbreaks.
\newblock \emph{PLOS ONE}, 5\penalty0 (9):\penalty0 1--8, 09 2010.
\newblock \doi{10.1371/journal.pone.0012948}.
\newblock URL \url{https://doi.org/10.1371/journal.pone.0012948}.

\bibitem[Cohen et~al.(2003)Cohen, Havlin, and ben
  Avraham]{PhysRevLett.91.247901}
Reuven Cohen, Shlomo Havlin, and Daniel ben Avraham.
\newblock Efficient immunization strategies for computer networks and
  populations.
\newblock \emph{Phys. Rev. Lett.}, 91:\penalty0 247901, Dec 2003.
\newblock \doi{10.1103/PhysRevLett.91.247901}.
\newblock URL \url{https://link.aps.org/doi/10.1103/PhysRevLett.91.247901}.

\bibitem[Euler(1741)]{euler1741solutio}
Leonhard Euler.
\newblock Solutio problematis ad geometriam situs pertinentis.
\newblock \emph{Commentarii academiae scientiarum Petropolitanae}, pages
  128--140, 1741.

\bibitem[Feld(1991)]{friends1}
Scott~L. Feld.
\newblock Why your friends have more friends than you do.
\newblock \emph{American Journal of Sociology}, 96\penalty0 (6):\penalty0
  1464--1477, 1991.
\newblock ISSN 00029602, 15375390.
\newblock URL \url{http://www.jstor.org/stable/2781907}.

\bibitem[Ghoshdastidar and von Luxburg(2018)]{ghoshdastidar2018practical}
Debarghya Ghoshdastidar and Ulrike von Luxburg.
\newblock Practical methods for graph two-sample testing.
\newblock In \emph{Advances in Neural Information Processing Systems}, pages
  3019--3028, 2018.

\bibitem[Girvan and Newman(2002)]{girvan2002community}
M.~Girvan and Mark E.~J. Newman.
\newblock Community structure in social and biological networks.
\newblock \emph{Proceedings of the National Academy of Sciences}, 99\penalty0
  (12):\penalty0 7821--7826, 2002.

\bibitem[Handcock et~al.(2007)Handcock, Raftery, and
  Tantrum]{handcock2007model}
Mark~S Handcock, Adrian~E Raftery, and Jeremy~M Tantrum.
\newblock Model-based clustering for social networks.
\newblock \emph{Journal of the Royal Statistical Society: Series A},
  170:\penalty0 301--354, 2007.

\bibitem[Hoff et~al.(2002)Hoff, Raftery, and Handcock]{hoff2002latent}
Peter~D Hoff, Adrian~E Raftery, and Mark~S Handcock.
\newblock Latent space approaches to social network analysis.
\newblock \emph{Journal of the American Statistical Association}, 97\penalty0
  (460):\penalty0 1090--1098, 2002.

\bibitem[Keeling(2005)]{keeling2005implications}
Matt Keeling.
\newblock The implications of network structure for epidemic dynamics.
\newblock \emph{Theoretical Population Biology}, 67\penalty0 (1):\penalty0
  1--8, 2005.
\newblock ISSN 0040-5809.
\newblock \doi{https://doi.org/10.1016/j.tpb.2004.08.002}.

\bibitem[Komolafe et~al.(2017)Komolafe, Quevedo, Sengupta, and
  Woodall]{komolafe2017statistical}
Tomilayo Komolafe, A~Valeria Quevedo, Srijan Sengupta, and William~H Woodall.
\newblock Statistical evaluation of spectral methods for anomaly detection in
  networks.
\newblock \emph{arXiv preprint arXiv:1711.01378}, 2017.

\bibitem[Kramer et~al.(2016)Kramer, Pulliam, Alexander, Park, Rohani, and
  Drake]{kramer2016spatial}
Andrew~M. Kramer, J.~Tomlin Pulliam, Laura~W. Alexander, Andrew~W. Park, Pejman
  Rohani, and John~M. Drake.
\newblock Spatial spread of the west africa ebola epidemic.
\newblock \emph{Royal Society Open Science}, 3\penalty0 (8):\penalty0 160294,
  2016.
\newblock \doi{10.1098/rsos.160294}.

\bibitem[Krivitsky et~al.(2009)Krivitsky, Handcock, Raftery, and
  Hoff]{krivitsky2009representing}
Pavel~N Krivitsky, Mark~S Handcock, Adrian~E Raftery, and Peter~D Hoff.
\newblock Representing degree distributions, clustering, and homophily in
  social networks with latent cluster random effects models.
\newblock \emph{Social Networks}, 31\penalty0 (3):\penalty0 204--213, 2009.

\bibitem[Leitch et~al.(2019)Leitch, Alexander, and Sengupta]{leitch2019toward}
Jack Leitch, Kathleen~A Alexander, and Srijan Sengupta.
\newblock Toward epidemic thresholds on temporal networks: a review and open
  questions.
\newblock \emph{Applied Network Science}, 4\penalty0 (1):\penalty0 105, 2019.

\bibitem[Leskovec et~al.(2007)Leskovec, Krause, Guestrin, Faloutsos,
  VanBriesen, and Glance]{leskovec2007cost}
Jure Leskovec, Andreas Krause, Carlos Guestrin, Christos Faloutsos, Jeanne
  VanBriesen, and Natalie Glance.
\newblock Cost-effective outbreak detection in networks.
\newblock In \emph{Proceedings of the 13th ACM SIGKDD international conference
  on Knowledge discovery and data mining}, pages 420--429, 2007.

\bibitem[Nadini et~al.(2018)Nadini, Sun, Ubaldi, Starnini, Rizzo, and
  Perra]{Nadini2018}
Matthieu Nadini, Kaiyuan Sun, Enrico Ubaldi, Michele Starnini, Alessandro
  Rizzo, and Nicola Perra.
\newblock {Epidemic spreading in modular time-varying networks}.
\newblock \emph{Scientific Reports}, 8\penalty0 (1):\penalty0 2352, 12 2018.
\newblock ISSN 2045-2322.
\newblock \doi{10.1038/s41598-018-20908-x}.

\bibitem[Ozella et~al.(2021)Ozella, Paolotti, Lichand, Rodr{\'\i}guez, Haenni,
  Phuka, Leal-Neto, and Cattuto]{ozella2021using}
Laura Ozella, Daniela Paolotti, Guilherme Lichand, Jorge~P Rodr{\'\i}guez,
  Simon Haenni, John Phuka, Onicio~B Leal-Neto, and Ciro Cattuto.
\newblock Using wearable proximity sensors to characterize social contact
  patterns in a village of rural malawi.
\newblock \emph{EPJ Data Science}, 10\penalty0 (1):\penalty0 46, 2021.

\bibitem[Prakash et~al.(2010)Prakash, Chakrabarti, Faloutsos, Valler, and
  Faloutsos]{Prakash2010}
B.~Aditya Prakash, Deepayan Chakrabarti, Michalis Faloutsos, Nicholas Valler,
  and Christos Faloutsos.
\newblock {Got the Flu (or Mumps)? Check the Eigenvalue!}
\newblock \emph{arXiv preprint arXiv:1004.0060}, 2010.

\bibitem[Rohe et~al.(2011)Rohe, Chatterjee, and Yu]{rohe2011spectral}
K.~Rohe, S.~Chatterjee, and B.~Yu.
\newblock Spectral clustering and the high-dimensional stochastic blockmodel.
\newblock \emph{The Annals of Statistics}, 39\penalty0 (4):\penalty0
  1878--1915, 2011.

\bibitem[Sengupta(2018)]{sengupta2018anomaly}
Srijan Sengupta.
\newblock Anomaly detection in static networks using egonets.
\newblock \emph{arXiv preprint arXiv:1807.08925}, 2018.

\bibitem[Shao et~al.(2016)Shao, Hossain, Wu, Khan, Vullikanti, Prakash,
  Marathe, and Ramakrishnan]{shao2016forecasting}
Huijuan Shao, KSM Hossain, Hao Wu, Maleq Khan, Anil Vullikanti, B~Aditya
  Prakash, Madhav Marathe, and Naren Ramakrishnan.
\newblock Forecasting the flu: designing social network sensors for epidemics.
\newblock \emph{arXiv preprint arXiv:1602.06866}, 2016.

\bibitem[Sueur et~al.(2012)Sueur, Deneubourg, and Petit]{sueur2012social}
C{\'e}dric Sueur, Jean-Louis Deneubourg, and Odile Petit.
\newblock From social network (centralized vs. decentralized) to collective
  decision-making (unshared vs. shared consensus).
\newblock \emph{PLoS one}, 7\penalty0 (2):\penalty0 e32566, 2012.

\bibitem[Tang et~al.(2017{\natexlab{a}})Tang, Athreya, Sussman, Lyzinski, Park,
  and Priebe]{tang2017semiparametric}
Minh Tang, Avanti Athreya, Daniel~L Sussman, Vince Lyzinski, Youngser Park, and
  Carey~E Priebe.
\newblock A semiparametric two-sample hypothesis testing problem for random
  graphs.
\newblock \emph{Journal of Computational and Graphical Statistics}, 26\penalty0
  (2):\penalty0 344--354, 2017{\natexlab{a}}.

\bibitem[Tang et~al.(2017{\natexlab{b}})Tang, Athreya, Sussman, Lyzinski, and
  Priebe]{tang2017nonparametric}
Minh Tang, Avanti Athreya, Daniel~L Sussman, Vince Lyzinski, and Carey~E
  Priebe.
\newblock A nonparametric two-sample hypothesis testing problem for random
  graphs.
\newblock \emph{Bernoulli}, 23\penalty0 (3):\penalty0 1599--1630,
  2017{\natexlab{b}}.

\bibitem[Wang et~al.(2003)Wang, Chakrabarti, Wang, and
  Faloutsos]{Wang2003EpidemicViewpoint}
Yang Wang, Deepayan Chakrabarti, Chenxi Wang, and Christos Faloutsos.
\newblock {Epidemic spreading in real networks: an eigenvalue viewpoint}.
\newblock In \emph{22nd International Symposium on Reliable Distributed
  Systems, 2003. Proceedings.}, pages 25--34, Florence, 2003. IEEE Comput. Soc.
\newblock ISBN 0-7695-1955-5.
\newblock \doi{10.1109/RELDIS.2003.1238052}.

\bibitem[Wang and Bickel(2017)]{wang2017likelihood}
YX~Rachel Wang and Peter~J Bickel.
\newblock Likelihood-based model selection for stochastic block models.
\newblock \emph{The Annals of Statistics}, 45\penalty0 (2):\penalty0 500--528,
  2017.

\bibitem[Watts and Strogatz(1998)]{watts1998collective}
Duncan~J Watts and Steven~H Strogatz.
\newblock Collective dynamics of `small-world' networks.
\newblock \emph{Nature}, 393:\penalty0 440--442, 1998.

\bibitem[Yan et~al.(2014)Yan, Shalizi, Jensen, Krzakala, Moore, Zdeborov{\'a},
  Zhang, and Zhu]{yan2014model}
Xiaoran Yan, Cosma Shalizi, Jacob~E Jensen, Florent Krzakala, Cristopher Moore,
  Lenka Zdeborov{\'a}, Pan Zhang, and Yaojia Zhu.
\newblock Model selection for degree-corrected block models.
\newblock \emph{Journal of Statistical Mechanics: Theory and Experiment},
  2014\penalty0 (5):\penalty0 P05007, 2014.

\bibitem[Zhao et~al.(2018)Zhao, Driscoll, Sengupta, Fricker~Jr, Spitzner, and
  Woodall]{zhao2018performance}
Meng~J. Zhao, Anne~R. Driscoll, Srijan Sengupta, Ronald~D. Fricker~Jr, Dan~J.
  Spitzner, and William~H. Woodall.
\newblock Performance evaluation of social network anomaly detection using a
  moving window–based scan method.
\newblock \emph{Quality and Reliability Engineering International}, 34\penalty0
  (8):\penalty0 1699--1716, 2018.

\bibitem[Zhao et~al.(2012)Zhao, Levina, and Zhu]{zhao2012consistency}
Yunpeng Zhao, Elizaveta Levina, and Ji~Zhu.
\newblock Consistency of community detection in networks under degree-corrected
  stochastic block models.
\newblock \emph{The Annals of Statistics}, 40:\penalty0 2266--2292, 2012.

\end{thebibliography}

\end{document}